\documentclass[a4paper,english]{elsarticle}
\usepackage[T1]{fontenc}
\usepackage[latin1]{inputenc}

\usepackage{url}

\usepackage{graphicx}
\newcommand{\eps}{\varepsilon}

\newcommand{\R}{\mathsf{R}}

\newcommand{\E}{I\!\!E}
\newsavebox{\fmbox}

\newtheorem{lemma}{Lemma}
\newtheorem{theorem}{Theorem}
\newtheorem{example}{Example}

\newtheorem{definition}{Definition}

\newenvironment{proof}{\noindent{\bf Proof : }\small }{\normalsize}

\begin{document}

\begin{frontmatter}

\title{Approximate Integration  of  streaming data}

\author{
        Michel de Rougemont}

\address{
    University Paris II\ \& IRIF-CNRS}
  
 \author{      Guillaume Vimont}

\address{
    University Paris II\ \& IRIF-CNRS}

%\layout

% DEBUT DE L'ARTICLE
%
\begin{abstract}
We approximate analytic queries on streaming data  with a weighted reservoir sampling. For a stream of tuples of a   Datawarehouse we show how to approximate some OLAP queries. For a stream of graph edges from a Social Network, we approximate the communities as the large connected components of the edges in the reservoir. We show that for a model of random graphs which follow a power law degree distribution, the community detection algorithm is a good approximation. Given two streams of graph edges from two Sources, we  define the {\em Community Correlation} as the  fraction of the nodes in communities in both streams. Although we do not store the edges of the streams, we can approximate  the Community Correlation and define the {\em Integration of two streams}. We illustrate this approach with Twitter streams, associated with  TV programs.\\\\
\end{abstract}

\begin{keyword}
	Streaming Algorithms, Data Integration,
Approximation, Complexity



\end{keyword}

\end{frontmatter}
\newpage

\tableofcontents

\newpage

\section{Introduction}
The integration of several Sources of data is also called the composition problem, in particular  when the Sources do not follow the same schema. It can be asked for two distinct Datawarehouses, two Social networks, or one Social network and one Datawarehouse. We specifically study the case of two streams  of labeled graphs from a Social network and develop several tools using   randomized streaming algorithms. We define several correlations  between two streaming graphs built from sequences of edges and study how to approximate them.

The basis of our approach is the approximation of analytical queries, in particular when we deal with streaming data. In the case of a Datawarehouse, we may  have a stream of tuples $t$ following an OLAP schema, where each tuple has a measure, and we may want to approximate OLAP queries. In the case of a Social network such as Twitter, we have a stream of tweets which generate edges of an evolving graph,  and we want to approximate the evolution of the communities as a function of time.

The main randomized technique used is a {\em $k$-weighted reservoir sampling} which maps an arbitrarly large stream of tuples $t$ of a Datawarehouse to $k$ tuples whose  weight is the measure $t.M$ of the tuple. It  also maps a stream of edges $u$ of a graph,   to $k$ edges and in this case the 
 measure of the edges is $1$.  We will show how we can approximate some OLAP queries and the main study will be the approximate dynamic community detection for graphs, using only the reservoir.  We store  the nodes of the graph in a database, but we do not store the edges. At any given time, we maintain the reservoir with $k$ random edges and compute the connected components of these edges.
We interpret the large connected components as communities and follow their evolution in time.

Edges of the reservoir are taken with a uniform distribution over the edges, hence the nodes of the edges are taken with a probability  proportional to their degrees. Random graphs observed in social networks often follow a power law degree distribution and random edges are  likely to connect nodes of high degrees. Therefore, the connected components of the random edges are likely to occur in  the dense subgraphs, i.e. in the communities. We propose a formal model of random graphs which follows a power law degree distribution with $p$ communities and will quantify the quality of the approximation of the communities.
 
A finite stream $s$ of edges can then be {\em compressed} in two parts: first the set $V$ of nodes stored in a classical database, and then the communities, i.e. sets $C_1,..C_l$ of size greater then a threshold $h$, at times $\tau, 2.\tau,....$ for some
constant $\tau$.  Given two finite streams $s_1,s_2$, the {\em node correlation} $\rho_V $ is the proportion of nodes in common and the {\em edge correlation}  $\rho_E $ is the proportion of edges connecting common nodes. 

We introduce the {\em community correlation } $\rho_C $ as the proportion of nodes in both communities among the common nodes. In our model, we  compute the node correlation, approximate the community correlation, but cannot compute the edge correlation as we do not store the edges. This new parameter can enrich the models of value associated with analytical queries such as the ones presented in 
\cite{rv2015} or in   \cite{EK2010}  for general mechanisms. 

The integration of two streams of edges defining two graphs $G_i=(V_i,E_i)$ for $i=1,2$ can then be viewed as the new structure
$$H=(V_1, V_2, V_1\cap V_2, C_1^1,..C_l^1,C_1^2,..C_p^2,\rho_C)$$ without edges, where $C_i^j$ is the $i$-th community of $G_j$ and $\rho_C$ is the Community Correlation.
All the sets are exactly or approximately computed from the streams with a database for $V$ and a finite memory, the size of the reservoir for the edges.

Our main application is the analysis of Twitter streams: a stream of graph edges for which we apply our $k$-reservoir. We temporarily  store a random subgraph  $\widehat{G}$ with $k$-edges and only store the large connected components of  $\widehat{G}$, i.e. of size greater than $h$ and their evolution in time. We give examples from the analysis of streams associated with TV shows on French Television (\#ONPC) and their  correlation.

Our main results are:
\begin{itemize}
\item  An approximation algorithm of simple OLAP queries for a Datawarehouse stream.

\item  An approximation algorithm for the community detection for graphs following a degree power law with a concentration,
\item A concrete analysis on Twitter streams to illustrate the model, and the community correlation
of Twitter streams.
\end{itemize}

We review the main concepts in section 2. We study the approximation of OLAP queries in a stream in section 3. In section 4, we consider streams of edges in a graph and give an approximate algorithm for the detection of communities. In section 5, we define the integration of streams and explain our experiments in section 6.

\section{Preliminaries}
The introduce our notations for OLAP queries and Social Networks, and the notion of approximation used.

\subsection{Datawarehouses and OLAP queries}
A Datawarehouse $I$ is a large table storing tuples $t$ with many attributes $A_1,...A_m,M$, some $A_i$ being foreign keys to other
tables, and $M$ a measure. Some auxiliary tables provide additional attributes for the foreign keys. An OLAP or star schema is a tree where each node is a set of attributes, the root is the set of all the attributes of $t$, and an edge exists if there is a functional dependency between the attributes of the origin node and the attributes of the extremity node. The {\em measure} is a specific node at depth 1 from the root.
An OLAP query for a schema $S$ is determined by: a filter condition, a {\em measure}, the selection of dimensions or classifiers, $C_1,...C_p$ where each $C_i$ is a node of the schema $S$, and an aggregation operator (COUNT, SUM, AVG, ...).

A filter  selects a subset of the tuples of the Datawarehouse, and we assume for simplicity
that SUM is the Aggregation Operator. The answer to an OLAP query is a multidimensional array, along the dimensions $C_1,...C_p$ and the {\em measure} $M$.  Each tuple
$c_1,...,c_p, m_i$ of the answer where $c_i \in C_i$  is such that $m_i=\frac{\sum_{t: t.C_1=c_1,...t.C_p=c_p} t.M}{ \sum_{t\in I} t.M}$. We consider relative {\em measures} as answers to OLAP queries and write $Q_C ^I$ as the distribution or
 density vector for the answer to  $Q$ on dimension $C$ and on  data warehouse $I$, as in Figure \ref{Pie_chart}.

\begin{example}
Consider tuples $t(${\em ID, Tags, RT, Time, User, SA)} storing some information about Twitter tweets. 
Let {\em Content}=\{Tags, RT\} where Tags is the set of tags of the Tweet and RT=1 if the tweet is a ReTweet and RT=0 otherwise. The measure $t.SA$ is  the Sentiment Analysis of the tweet, an integer value in $[1,2,...10]$. The sentiment is negative if $SA<5$ and positive when $SA\geq 5$ with a maximum of $10$.
The simple OLAP schema of Figure \ref{schema}  describes the possible dimensions and the measure $SA$.
The edges indicate a functional dependency between sets of attributes.

\begin{figure}[ht]\label{schema}
\centering
 \includegraphics[width=10cm]{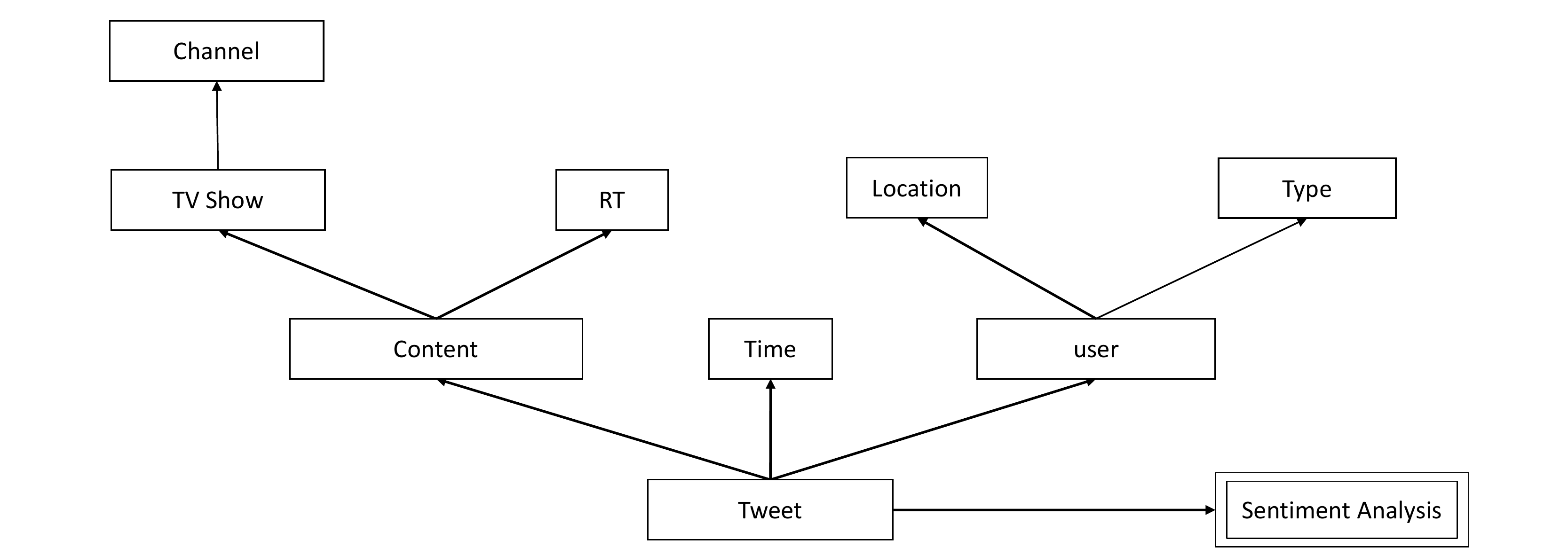}
\caption{An OLAP schema for a Datawarehouse storing tuples $t$ for each Twitter tweet, with {\em Sentiment Analysis}, an integer in  $[1,2,...10]$ as a measure.}
\end{figure}

Consider the analysis on the dimension {\em C=Channel}, with two possible values $c$ in the set  \{CNN, PBS\}.
The result is a distribution $Q_{C}$ with $Q_{C=CNN} ^I=2/3$ and $Q_{C=PBS} ^I=1/3$ as in Figure \ref{Pie_chart} . The approximation of $Q_{C}$ 
is studied in section 3.
In this case $\mid  C \mid=2$, i.e.  $\mid  C \mid$ is the number of values of the dimension $C$.

\begin{figure}[ht]
\begin{center}
 \includegraphics[width=7cm]{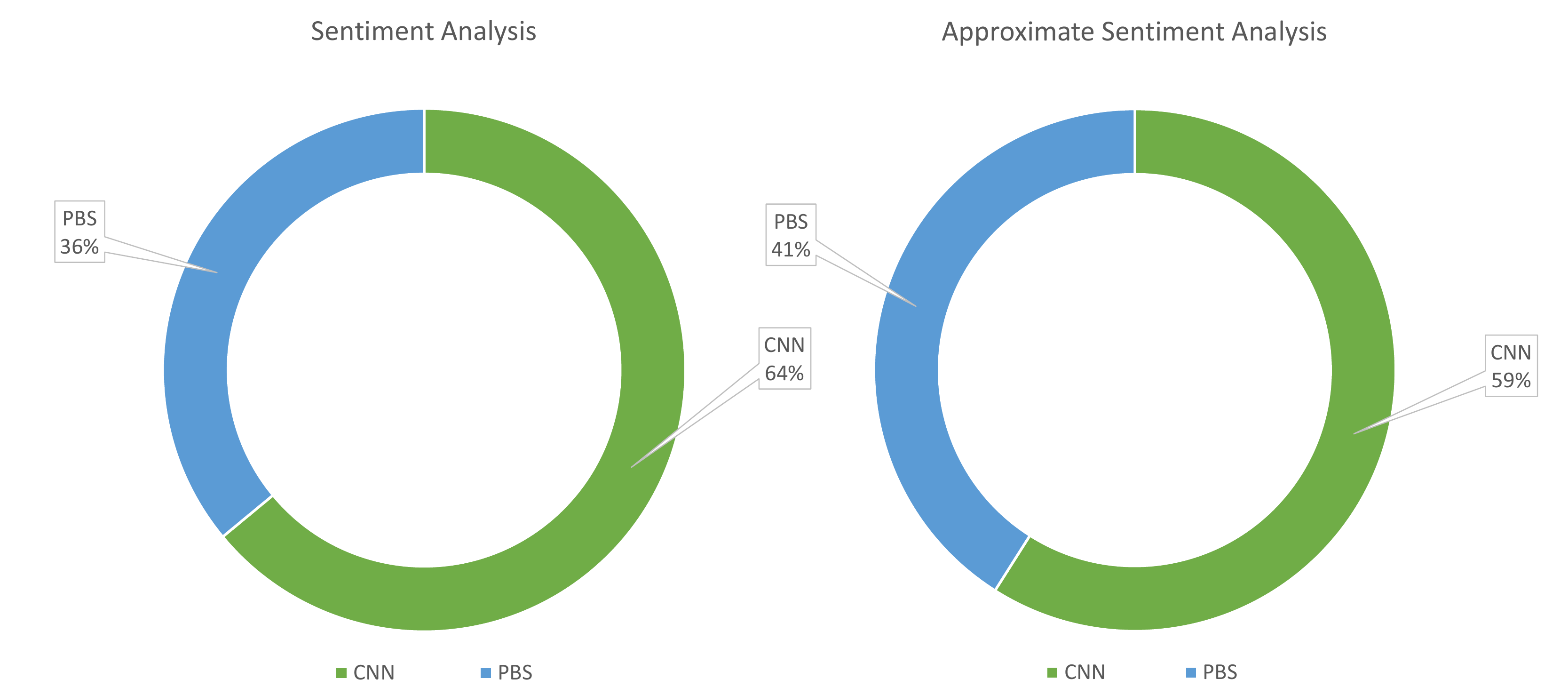}
 \caption{An OLAP query for  the Sentiment Analysis per Channel. 
The exact solution $Q_{C=CNN} ^I=0.66$ and the approximate solution $Q_{C=CNN} ^I=0.61$ with a reservoir.}\label{Pie_chart}
\end{center}
\end{figure}
\end{example}

\subsection{Social Networks}
A social network is a labeled graph $G=(V,E)$ with domain $V$ and edges $E\subseteq V.V$. 
In many cases, it is built as a stream of edges $e_1,.....e_m$ wich define $E$.
Given a set of tags, Twitter provides
 a stream of tweets represented as Json trees. We construct the {\em Twitter Graph} of the stream, i.e. the graph $G=(V,E)$ with multiple edges $E$
where $V$ is the set of  tags ($\#x$ or $@y$ ) seen and  for each tweet sent by  $@y$  which contains tags
$\#x$ ,$@z$ we construct the
edges  $(@y,\#x)$ and  $(@y,@z)$   in $E$.

Social Networks   graphs  have a specific structure. The graphs are mostly connected, the degree distribution  of the nodes follows a power law and the {\em communities } are defined as the dense subgraphs. The detection of communities is a classical problem, viewed by many techniques such as Mincuts, hierarchical clustering or the Girwan-Newman algorithm based on the edge connectivity. All these methods require to store the whole set of edges.

By contrast, we will detect communities without storing the edges, from the stream of edges, and approximate the dynamic of the communities. We will
also use this technique to compress a stream and to  integrate two streams.

\subsection{Approximation}

In our context, we approximate density values less than $1$  of the OLAP queries or  communities of a graph.
We use randomized algorithms
with an  additive approximation, and the probabilistic space $\Omega$ for a stream $s$ of $m$ tuples (resp. edges) is a subset of $k$ tuples (resp. edges) where each edge occurs with some probability $p$. In the case of edges, the probability $p$ is uniform, i.e. $p=1/m$.
 There are usually two parameters $0 \leq \eps, \delta \leq 1$ for the approximation of randomized algorithms, where $\eps$ is the error, and $1-\delta$ the confidence.
 
In the case of the density value, i.e. a function $F:\Sigma^* \rightarrow \R$ where $\Sigma$ is the set of possible tuples, let $A$ be a randomized  algorithm with input $s$ and output 
$y=A(s)$ where $y\in \R$ is the density value.
 The  algorithm $A(s)$ will  $(\epsilon,\delta)$-approximate the  function $F$ if for
  all $s$,
\[Prob_{ \Omega } [ F(s)-\eps \leq A(s) \leq F(s)+\eps] ~ \geq ~ 1-\delta \] 
In the case of a density vector $Q$, we use the $L_1$ distance between vectors. The algorithm $A(s)$  approximates $Q$ if
$Prob_{ \Omega } [ \mid Q- A(s) \mid_1 \leq \eps] ~ \geq ~ 1-\delta $.
The  randomized algorithm $A$ takes samples $t\in I$ from the stream with different distributions, introduced in the next subsection and in section 3.\\

In the case of the community detection, it is   important  to detect a community $S\subseteq V$ in a graph $G=(V,E)$ with a set $C\subseteq V$ which intersects $S$. The function $F:\Sigma^* \rightarrow 2^V$ takes a stream $s$ of edges as input and $F(s) \subseteq V$.
The  algorithm $A$  $\delta$-approximates the  function $F$ if for
  all $s$,
\[Prob_{ \Omega } [ A(s) \cap F(s) \neq \emptyset] ~ \geq ~ 1-\delta \] 

The randomized algorithm $A$ takes sample edges from the stream $s$ with a uniform distribution and outputs a subset $A(s)=C$ of the nodes.
If there is no output then $A(s)=\emptyset$. Approximate algorithms for streaming data  are studied in \cite{M2005}, with a particular emphasis on the space required. The algorithms presented require a space of $\mid V \mid + k.\log \mid V \mid $.

\subsubsection{Reservoir Sampling}
A classical technique, introduced in \cite{V85} is to sample each new tuple (edge) of a stream $s$ with some probability $p$ and to keep it in a set $S$ called the 
{\em reservoir} which holds $k$ tuples.  In the case of tuples $t$ of a Datawarehouse with a measure $t.M$, we keep them  with a probability proportional to their measures.

Let $s=t_1,t_2,....t_n$ be the stream of tuples $t$ with the measure $t.M$, and let $T_n=\sum_{i=1,...n} t_i.M$ and let $\widehat{S_n}$ be the reservoir at stage $n$. We write  $\widehat{S}$ to denote that $S$ is a  random variable.\\

{\bf $k$-Reservoir sampling: A(s)}
\begin{itemize}
\item Initialize $S_k=\{t_1,t_2,....t_k\}$,
\item For $j=k+1,....n$, select $t_j$ with probability  $(k*t_j.M)/T_j$. If it is selected replace a random element of the reservoir (with probability $1/k$) by $t_j$.\\
\end{itemize}

The key property is that each tuple $t_i$ is taken proportionally to its measure.
It is a classical simple argument which we recall.

\begin{lemma}
Let $S_n$ be the reservoir at stage $n$. Then for all $n>k$ and $1\leq i \leq n$:
$$Prob[  t_i \in S_n] = k.t_i.M/T_n]$$

\end{lemma}
\begin{proof}
Let us prove by induction on $n$. The probability at stage $n+1$ that $t_i$ is in the reservoir
$Prob[  t_i \in S_{n+1}]$ is composed of two events: either the tuple $t_{n+1}$ does not enter the reservoir, with probability $(1-k.t_{n+1}/T_{n+1})$ or the tuple $t_{n+1}$ enters the reservoir with probability $k.t_{n+1}/T_{n+1}$ and the tuple $t_i$  is maintained with probability $(k-1)/k$.
Hence:
$$Prob[  t_i \in S_{n+1}] = k.t_i.M/T_n (  (1-k.t_{n+1}/T_{n+1}) +   k.t_{n+1}/T_{n+1} ~.                                      (k-1)/k )$$
$$Prob[  t_i \in S_{n+1}] = k.t_i.M/T_n (  1-t_{n+1}/T_{n+1})= k.t_i.M/T_{n+1}$$
\end{proof}

In the case of edges, the measure is always $1$ and all the edges are uniform.

\section{Streaming Datawarehouse and approximate OLAP}

\label{samplingAlgorithm}

Two important methods can be used to sample a Datawarehouse stream $I$:\\
\begin{itemize}
\item Uniform sampling: we select $\widehat{I}$, made of $k$ distinct samples of $I$, with a uniform reservoir sampling on the $m$ tuples,
\item Weighted sampling:  we select $\widehat{I}$ made of $k$ distinct samples of $I$, with a {\em $k$-weighted reservoir sampling} on the $m$ tuples. The measure of the samples is set to $1$.\\
\end{itemize}

We concentrate on a  $k$-weighted reservoir.
Let $ \widehat{Q_C}$ be the density of $Q_C$ on $\widehat{I}$ as represented in Figure \ref{Pie_chart}, with the weighted sampling, i.e.
$ \widehat{Q}_{C=c}$ be the density of $Q$ on the value $c$ of the dimension $C$, i.e. the number of samples such that  $C=c$ divided by $k$.
 The algorithm $A(s)$ simply interprets the samples with a  measure of  $1$, i.e. computes $ \widehat{Q}_{C}$.
 
 In order to show that  $ \widehat{Q}_{C}$  is an  $(\eps,\delta)$-approximation of $Q_{C}$, we look at each component $Q_{C=c}$. We show
 that $\E( \widehat{Q}_{C=c})$ the expected value of $ \widehat{Q}_{C=c}$  is  $Q_{C=c}$. We then apply a Chernoff bound and a union bound.
 
\begin{theorem}\label{l-OLAP}
$Q_{C}$, i.e. the density of $Q$ on the dimension $C$  can be $(\eps,\delta)$-approximated by 
$ \widehat{Q}_{C}$ if $k\geq   \frac{1}{2}.(\frac{\mid  C \mid}{\epsilon})^2 . \log \frac{1}{\delta}$.

\end{theorem}
\begin{proof}
Let us evaluate  $\E( \widehat{Q}_{C=c})$, the expectation of the density of the samples. It is the expected number of samples with $C=c$ divided by $k$ the total number of samples. The expected number of samples is  $\sum_{t: t.C=c} \frac{k.t.M}{T} $ as each $t$ such that $C=c$ is taken with probability $\frac{k.t.M}{T}$ by the weighted reservoir for any total weight $T$. Therefore:

$$\E( \widehat{Q}_{C=c})= \frac{ \sum_{t: t.C=c} \frac{k.t.M}{T}  }{k}= \frac{ \sum_{t: t.C=c} t.M}{T} =Q_{C=c}  $$
 i.e. the expectation of the density $\widehat{Q}_{C=c}$ is precisely $Q_{C=c}$.
 As  the tuples of the reservoir are taken {\em independently} and as the densities are less than $1$, we
 can apply a Chernoff-Hoeffding bound \cite{H63}:
$$Prob[\mid Q_{C=c}- \E( \widehat{Q}_{C=c}  )\mid \geq t] \leq e^{-2t^{2}.k}$$

In this form, $t$ is the error and $1-\delta=  1-e^{-2t^{2}.k}$ is the confidence.
We set $t= \frac{\epsilon   } {\mid  C \mid} $, and $\delta=e^{-2t^{2}.k}$. We apply the previous inequality for all $c\in C$. With a union bound, we  conclude that if $k >  \frac{1}{2}.(\frac{\mid  C \mid}{\epsilon})^2 . \log \frac{1}{\delta} $ then:

$$Prob[\mid Q_{C}- \E( \widehat{Q_{C}}  )\mid \leq \eps ] \geq 1- \delta$$
\end{proof}

This result generalizes to arbitrary dimensions but  is of limited use in practice. If the OLAP query has a selection $\sigma$, the result will not hold.
However if we sample on the stream after we apply the selection, it will hold again. Hence we need to combine sampling and composition operations in a non trivial way. 

In particular, if we combine two Datawarehouses with a new schema, it is difficult to correctly sample the two streams. In the case of two graphs, i.e. a simpler case, we propose a solution in the next section.

\section{Streaming graphs}

We consider a stream of edges $e_1, e_2,.....e_m$ which defines a family of graph $G_m=(V,E) $ at stage $m$ such that
$E=\{ e_1, e_2,.....e_m \}$ is on a domain $V$. In this case, the graphs are monotone as no edge is removed. In the {\em Window model},
we only consider the last edges, i.e. $e_{m-j}, e_{m-j+1},.....e_m$.  In this case some edges are removed and some edges are added to define a graph $G_w$. We will consider both models, when $j$ is specified by a time condition such as the last hour or the last 15mins.

In both models, we keep all vertices in a database but only a few random edges. We maintain a uniform reservoir sampling of size $k$ and consider the random $\widehat{G}$ defined by the reservoir, i.e. $k$ edges, when $G_m$ is large. Notice that in the reservoir, edges are removed and added hence $\widehat{G}$ is maintained as in
  the window model.
 In many Social Networks, the set of nodes $V$ is large but reaches a limit, whereas the set of edges is much larger and cannot be efficiently stored.

\subsection{Random graphs}

The most classical model of random graphs is the Erd\"{o}s-Renyi $G(n,p)$ model (see \cite{E60} ) where $V$ is a set of $n$ nodes and each edge $e=(i,j)$ is chosen independently with probability $p$. In the 
Preferential Attachment model, $PA(m)$, (see \cite{B99} , the random graph $\widehat{G}_n$ with $n$ nodes is built dynamically:  given  $\widehat{G}_n$  at stage $n$, we build  $\widehat{G}_{n+1}$ by adding a new node and $m$ edges connecting the new node with a random node $j$ following the degree distribution in $\widehat{G}_n$. The resulting graphs have a degree distribution which follows a power law, i.e.
$$Prob[ d(i)=j ] = \frac{c}{j^2}$$
when the node $i$ is selected uniformly.

In yet another model $D(\delta)$, we fix a degree distribution, $\delta=[D(1),D(2),....D(k)]$ where $D(i)$ is the number of nodes of degree $i$ and generate a random graph uniform among all the graphs with $\sum_i D(i)$ nodes and
$\sum_i i*D(i)/2$ edges.
For example if $\delta=[4,3,2]$\footnote{Alternatively, one may give a sequence of integers, the degrees of the various nodes in decreasing order, i.e. $[3,3,2,2,2,1,1,1,1]$, a sequence of length $9$ for the distribution $\delta=[4,3,2]$.}, i.e. approximately a power law, we search for a graph with $9$ nodes and $8$ edges. Specifically  $4$ nodes of degree $1$, $3$ nodes of degree $2$ and $2$ nodes of degree $3$, as in Figure \ref{crg} (a). Alternatively, we may represent $\delta$ as a distribution, i.e. $\delta=[\frac{4}{9},\frac{1}{3},\frac{2}{9}]$.

The configuration model  generates graphs with the distribution $\delta$ when $\sum_i i*D(i)$ is even. Enumerate the nodes with half-edges according to their degrees, and select a random matching between the half-edges. The graph may have multiple edges. If $\delta$ follows a power law, then the maximum degree is $O(\sqrt{m})$ if the graph has $m$ edges. 

A $D(\delta)$ graph is {\em concentrated } if  all the nodes of maximum degrees are densely
connected. It can be obtained if the matching  has a preference for nodes with high degrees, as
in Figure \ref{crg1}.

\begin{definition}
A $D(\delta)$ graph with $m$ edges is {\em concentrated } when $\delta$ follows a power law if
 the $O(\sqrt{m/2})$ nodes of highest degree form a  dense subgraph $S$, i.e. each node $i \in S$ has a majority of its neighbors in $S$.
\end{definition}

We will call $S$ the community of the concentrated graph $D(\delta)$. If a node is is of degree $3$ in $S$, then at least $2$ neighbors must be in $S$, if it is of degree $2$ in $S$, then at least $1$ neighbor must be in $S$. It can be checked for $S$ of size $3$ in Figure \ref{crg}.

\begin{figure}[ht]
\centering
 \includegraphics[width=8cm]{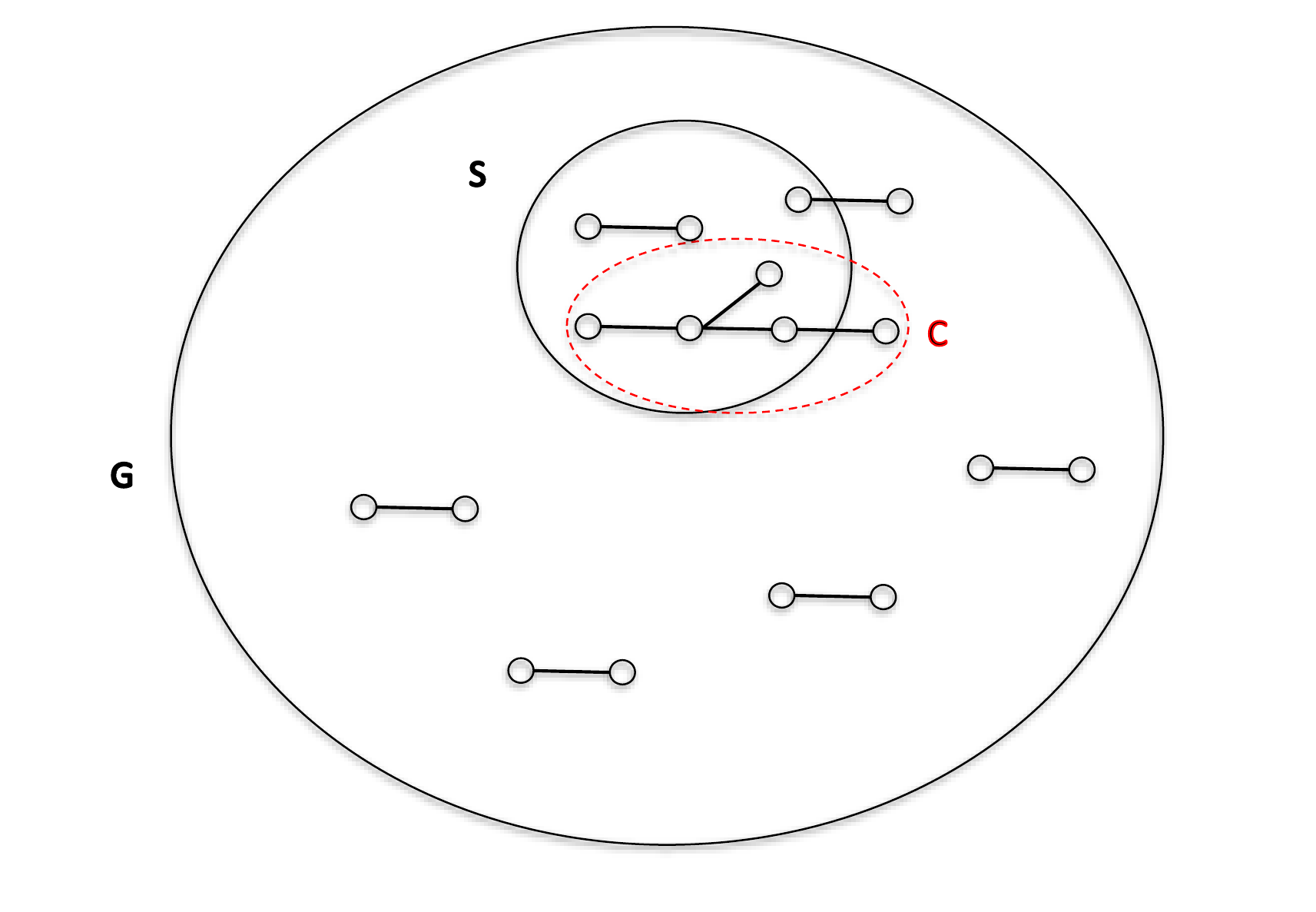}
\caption{Concentrated random graph $G$ with a community $S$ and $10$ random edges from the reservoir defining $\widehat{G}$ with the large connected component $\widehat{C}$ with $4$ edges.}\label{crg1}
\end{figure}

The set $S$ is close to a clique of size $O(\sqrt{m/2})=n'$ and edges are taken with probability 
$1/m$. We will show that the  probability that an edge is in the clique $S$ is $\alpha/m=p'$.
We are then close to the  Erd\"{o}s-Renyi $G(n',p')$ model where $p'=2.\alpha/n'^2$. In this regime,
we know from \cite{B2001} that the largest connected component is small, of order 
$O(\log n')=O(\log (\sqrt{m}))$.
The giant connected component requires $p' \geq (\log n')/n'$.
The size of the connected components in a graph  specified by a degree sequence is studied in \cite{C2002}.

\subsection{Random graphs with $p$ communities}

None of the previous models exhibit  many distinct community structures. The $PA(m)$ model or the power law distribution create only one dense community.  Consider  two random graphs $\widehat{G}_1$ and $\widehat{G}_2$  of the same size following the $D(\delta)$ model when  $\delta$ follows a power law.  We say that $\widehat{G}$ follows the $D(\delta)^2$ model if
$$\widehat{G}=\widehat{G}_1 \mid \widehat{G}_2$$
i.e. $\widehat{G}$ is the union of $\widehat{G}_1$ and  $\widehat{G}_2$ with a few random edges connecting the nodes of low degree. This construction exhibits two communities $S_1$ and $S_2$ and generalizes to $D(\delta)^p$ for $p$ communities of different sizes, as in
Figure \ref{crg}.

\begin{figure}[ht] 
\centering
 \includegraphics[width=9cm]{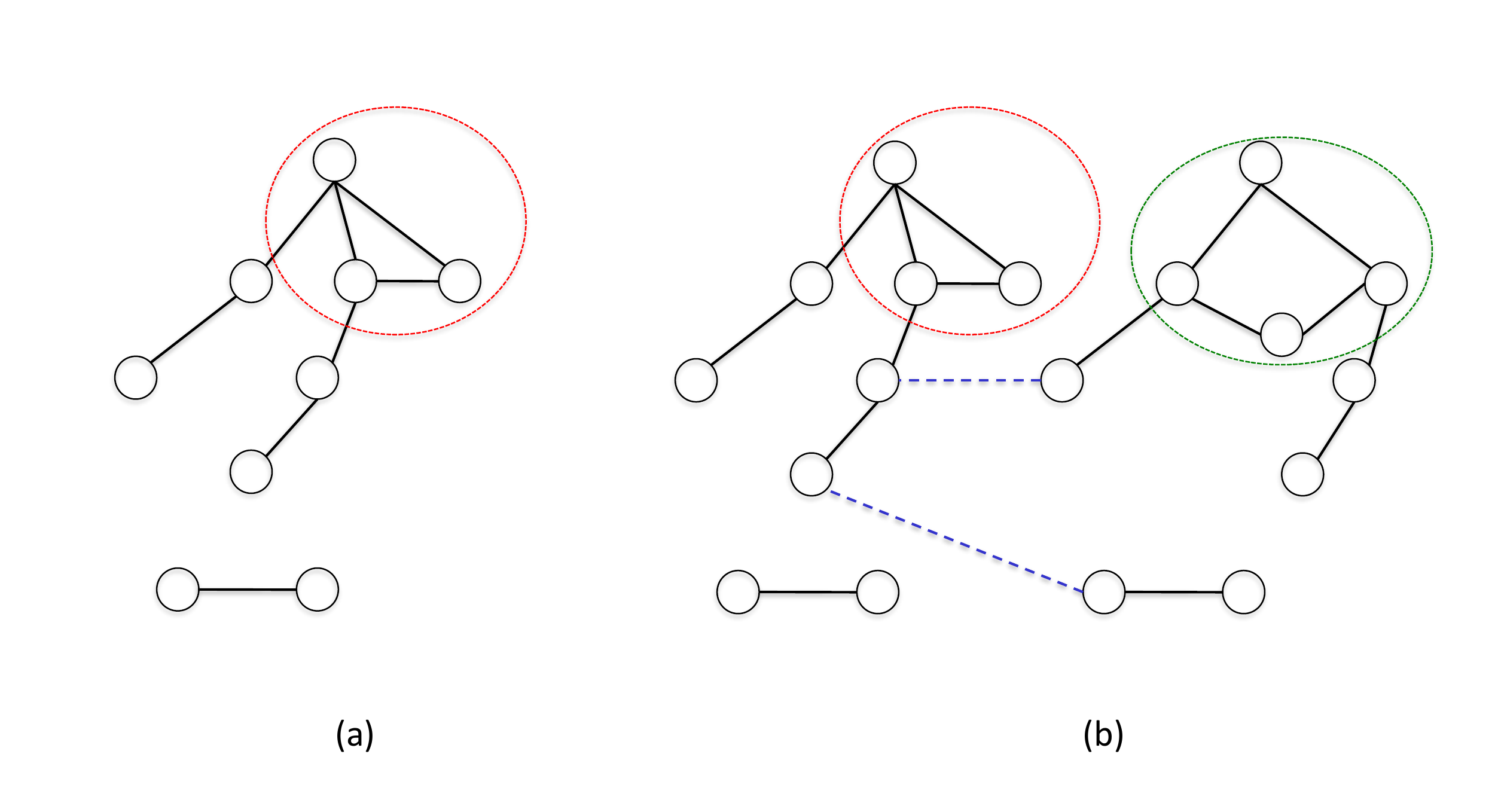}
\caption{Concentrated random graph for  $D(\delta)$ with one community in (a).  
Random graph for $D(\delta)^2$ with 2 communities  in (b) where $\delta=[4,3,2]$ (or $[\frac{4}{9},\frac{1}{3},\frac{2}{9}]$ as a distribution).}\label{crg}
\end{figure}

Notice that if $\widehat{G}_1$ and $\widehat{G}_2$ have the same size and the same degree distribution $\delta=[\frac{4}{9},\frac{1}{3},\frac{2}{9}]$, then $\widehat{G}=\widehat{G}_1 \mid \widehat{G}_2$ has approximately the same distribution $\delta$.

\subsection{Reservoir based random subgraphs}

We maintain a reservoir with $k$ edges, whose edges occur with probability  
$\frac{1}{m}$ in a stream with $m$ edges for any large $m$, i.e. edges are uniformly selected. Such random graphs are considered in \cite{L2006} in a different setting, under the name MST (Minimum Spanning tree) where an arbitrary random order is selected on the edges, hence  each edge is uniformly selected.

We can also  select nodes from a reservoir with $k$ edges, by choosing an edge $e=(i,j)$ and then choosing $i$ or $j$ with probability $\frac{1}{2}$. In this case, we select a node with probability proportional to its degree $d(i)$, simply because $d(i)$ independent edges connect to $i$.
Therefore, the reservoir magically selects edges and nodes with high degrees, even so we never store any information about the degree of the nodes. 

If we wish to keep only the last edges, for example the edges read in the last hour, the reservoir sampling will not guarantee a uniform distribution.
A priority sampling for the sliding window  \cite{M2014} assigns a random value in the $[0,1]$ interval to each edge    and selects the edge with minimum value. Each edge is selected with the uniform distribution.\\

\subsection{Community detection}
A graph has a community structure if the nodes can be grouped into $p$ dense subgraphs.
Given a graph $G=(V,E)$, we want to partition $V$ into $p+1$ components, such that
$V=V_1 \oplus V_2.... \oplus V_p \oplus V_{p+1}$ where each $V_i$ for $1\leq i \leq p$ is dense,
i.e. $|E_i|\geq \alpha. |V_i|^2$ for some constant $\alpha$, and $E_i$ is the set of edges connecting nodes of $V_i$. The set $V_{p+1}$ groups nodes which are not parts of the communities.

In the simplest case of $2$ components, $V=V_1 \oplus V_2 \oplus V_3$ and $V_1,V_2$ are dense and $V_3$ is the set of unclassified nodes, which can also be viewed as {\em noise}. If we want to approximate the communities, we want to capture most of the nodes of high degrees in $V_1$ and $V_2$. We adapt the definition and require that:
$[Prob_{ \Omega } [ A(s) \cap S_1  \neq \emptyset  \wedge  A(s) \cap S_2 \neq \emptyset] ~ \geq ~ 1-\delta$.\\

{\em Algorithm for  Community detection in a stream $s$ of $m$ edges $A(k,c,h)$:}
\begin{itemize}
\item Maintain a $k$-reservoir,
\item For each $c$ edges, update the nodes database and the large (of size greater than $h$) connected components $\widehat{C_1},...\widehat{C_l}$ of the $k$-reservoir window.  \\\\
\end{itemize}

 In practice $k=400$, $c=3, h=3$. Therefore each $\widehat{C_i}$ will contain nodes of high degrees, and we will interpret $\widehat{C_i}$  as a community at a time $t$. Figure \ref{connected} is an example of the connected components of the reservoir.
 
\begin{figure}[ht]
\begin{center}
 \includegraphics[width=8cm]{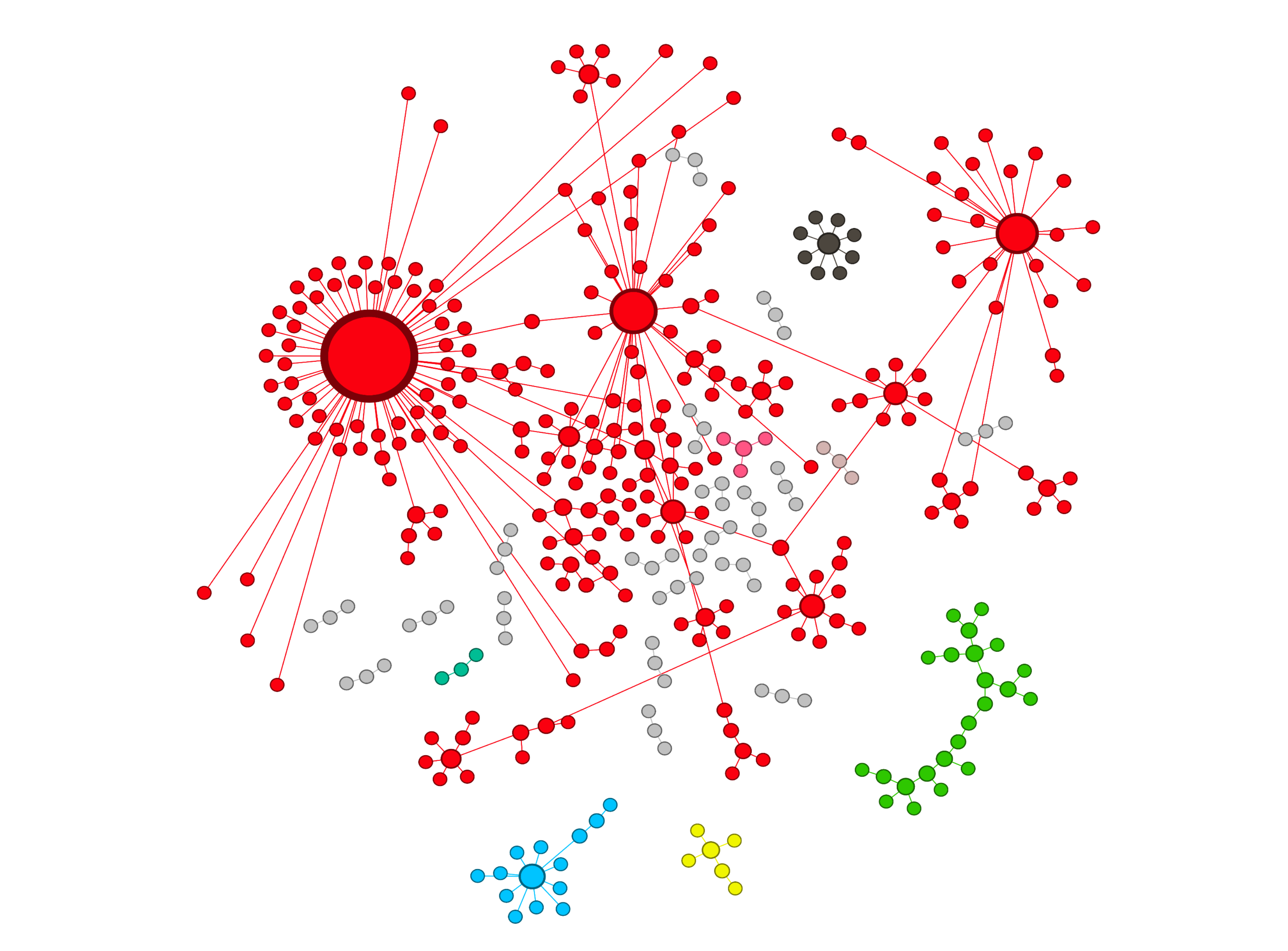}
 \caption{Connected components of the $k$-reservoir.}\label{connected}
\end{center}
\end{figure}

\begin{lemma}\label{l-graph}
Let $S$ be the community of a $D(\delta)$ graph following a power law, with $m$ edges. There are two constants $\alpha,\beta$, which depend on the distribution $\delta$ such that:
$$Prob[ e_i \in E_S] >\alpha$$
$$Prob[ e_i\in E_S \wedge e_j \in E_S \wedge e_i, e_j {\rm ~share ~a ~node}] > \beta$$
\end{lemma}
\begin{proof}
Recall that $S$ contains the $O(\sqrt{m/2})$ nodes of highest degree. The degrees are  from $O(\sqrt{m})$ until at least $O(\sqrt{m}-\sqrt{m/2})$.
Among the possible $m/4$ internal edges of $S$, we have a constant proportion because at least half of the edges coming from a node
must be internal. As a random edge $e_i$ is chosen with probability $1/m$, it has a constant probability to be internal, i.e. there exists $\alpha$ such  that:
$$Prob[ e_i \in E_S] >\alpha$$
 $S$ is dense , i.e. it contains a constant fraction $\alpha$ of the possible edges, hence a fraction$1-\alpha$ of pairs which are non-edges.
 If we select two independent edges $e_i,e_j$ they are internal with probability $\alpha ^2$. The probability that they share a node is
 $1-\eta$ if $\eta$ is the probability that they do not share a node. The probability that they do not share a node is the probability that some edge or some non-edge connects each of the $4$ nodes of $e_i,e_j$. There  are $4$ possible connecting edges, hence $16$ possibilities, but $\eta$ is bounded by a constant, hence $1-\eta$  is also constant. If we set: $\beta=\alpha ^2.(1-\eta)$, we obtain:
 $$Prob[ e_i\in E_S \wedge e_j \in E_S \wedge e_i, e_j {\rm ~share ~a ~node}] > \beta$$
\end{proof}
We can think of $\alpha$ as $1/4$ and $\beta=1/10$.
We can now prove the main result in the case of $p=2$ communities, i.e. $G=G_1 \mid G_2$, where the
graphs $G_1$ and $G_2$ have the same size. It generalizes to an arbitrary $p$ and to  graphs $G_i$ that do not have the same size. The size
 must be at least a fraction of $m$.\\

\begin{theorem}
Let $G$ be a  $D(\delta)^2$ graph following a power law, with $2m$ edges.
There exists a constant  $\delta$ such that the DC-Algorithm $\delta$-approximates the communities of $G=G_1 \mid G_2$. 
\end{theorem}
\begin{proof}
By applying Lemma \ref{l-graph}, we expect $k.\alpha.m/2$ edges in each dense component $S_1$ or $S_2$.
The other edges could have 
 one extremity in $S_i$ and the other in $V_i -S_i$ or both in $V_i -S_i$.
In each $V_i$ there may be several connected components. We consider the largest $\widehat{C}_1$ for $G_1$ and $\widehat{C}_2$ for $G_2$.
We need to estimate  the probability 
$$Prob[ \mid C_i \mid \geq h \wedge \widehat{C}_i \cap S_i \neq \emptyset ] $$
for $i=1,2$.
Using the same argument as the one used in Lemma \ref{l-graph}, there exists a $\gamma$ such that:\\
$Prob[ e_{i_1}\in E_S \wedge e_ {i_2}\in E_S ... \wedge e_ {i_h} \in E_S \wedge e_{i_1}, e_{i_2},...e_{i_h} {\rm ~are  ~connected}] > \gamma$.
We just evaluate the probability that there are not connected, i.e. one of the edges is not connected to the others because there exist edges and non edges to each of the nodes of the other edges. Hence if $\widehat{C}_i$ is the largest connected component in $S_i$:
$$Prob[ \mid \widehat{C}_i \mid \geq h] > \gamma$$
and if we take $\delta=\gamma^2$ we conclude that $Prob[ \mid \widehat{C}_1\mid \geq h  \wedge \mid \widehat{C}_2\mid \geq h] > \delta$.\\
\end{proof}

Clearly, if the number $p$ of components is large, the quality of the approximation  decreases. If the size of some communities is small,  the chance of not detecting it will also increase.

\subsection{Dynamic Community detection}

We extend the community detection algorithm and maintain two $k$-reservoirs: one for the global data, and one for most recent items.
A priority sampling \cite{M2014}, provides a uniform sampling of the last elements of the stream, defined by a time condition such as the last 15mins. We call it a {\em $k$-reservoir window}.

We update the  the connected components for every  $c$ new edges (for example $c=5$) in the stream. We store the connected components at regular time intervals.\\

{\em DC-Algorithm for Dynamic Community detection of a stream $s$ of edges: $DC(k,h,c,\tau)$}
\begin{itemize}
\item Maintain a global $k$-reservoir and a $k$-reservoir window,
\item For each $c$ edges, update the nodes database and the large (greater than $h$) connected components $\widehat{C_1},...\widehat{C_l}$ of the $k$-reservoir window.  When we remove edges, the components may split or disappear. When we add edges, components may merge or appear.
\item Store
the components of size greater than $h$ at some time interval $\tau$.
\item When the stream stops, store the global connected components $\widehat{C_g}_1,...\widehat{C_g}_l$ of the $k$-reservoir. 
\end{itemize}

\begin{figure}[ht]  
\begin{center}
 \includegraphics[width=10cm]{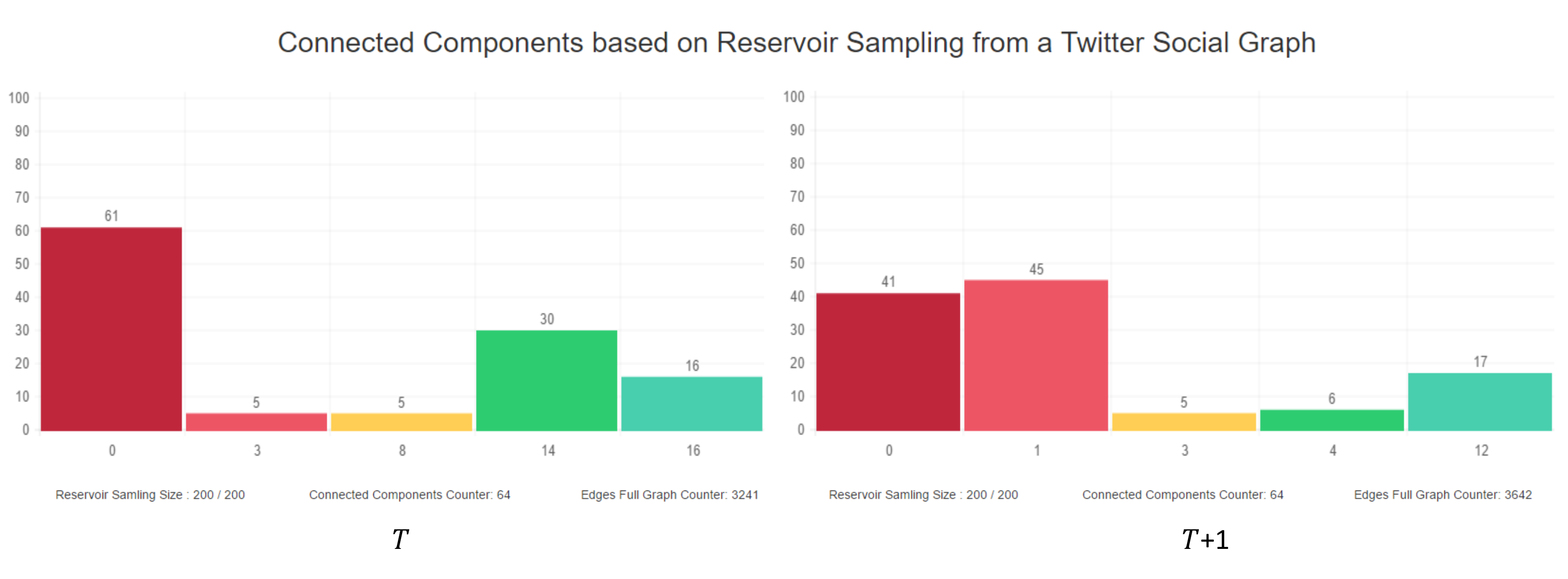}
 \caption{Sizes of the connected components online}\label{Graph_reservoir} 
\end{center}
\end{figure}

In the implementation, $k=400, h=3, c=5, \tau=15mins$. Figure \ref{Graph_reservoir} shows the dynamic evolution of the sizes of the communities between two iterations.

\subsection{Stability of the components}

As we observe the dynamic of the communities, there is some instability: some components appear, disappear and may reappear later.
It is best observed with the following experiment: assume two independent reservoirs of size $k'=k/2$ as in Figure \ref{Graph_2_reservoirs}.
The last two communities of the reservoir $1$  with $5$ communities merge to correspond to the $4$ communities in reservoir 2.

\begin{figure}[ht]  
\begin{center}
 \includegraphics[width=10cm]{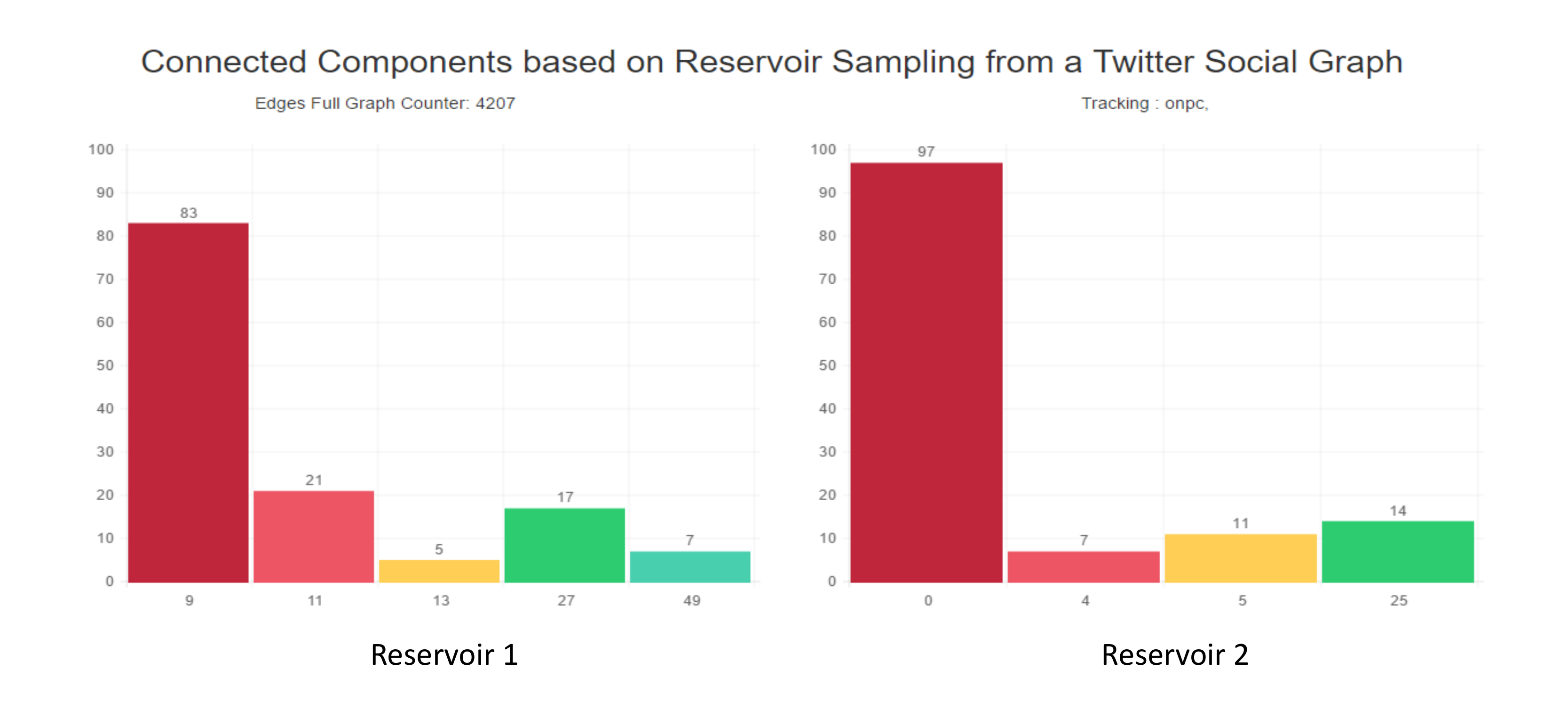}
 \caption{Sizes of the connected components with 2 independent reservoirs}\label{Graph_2_reservoirs}
\end{center}
\end{figure}

Consider the subgraph $G_i$ of the community $C_i$. It is most likely a tree if $C_i$ is small, hence unstable as 
the removal of $1$ edge splits the component  or makes it small and it disappears. Larger components are graphs
which are therefore more stable.
If the original graph with $m$ edges has a concentrated component $S$ of size $O(\sqrt{m/2})=n$, then we
can estimate with the
Erd\"{o}s-Renyi model $G(n,p)$ the connected components inside $S$. In this case $p=2.\alpha/n^2$ and we are in the sparse regime as $p<\log n/n$. The
 components are most likely trees of size at most $O(\log (\sqrt{m/2})$.  Hence the instability of the small components.

\section{Integration from multiple  sources}

Given two streams of edges defining two graphs $G_i=(V_i,E_i)$ for $i=1,2$, what is the integration of these two structures?
The {\em node correlation} and the {\em edge correlation} between two graphs $G_1$ and $G_2$ are:
$$\rho_V= \frac{|V_1 \cap V_2|}{max\{  |V_1|, |V_2|\}} ,    \rho_E=\frac{|E_1 \cap E_2|}{max\{  |E_1|, |E_2|\}} $$

As we store $V_1$ and $V_2$, we can compute $\rho_V$, but we cannot compute $\rho_E$, as we do not store
$E_1$ nor $E_2$.  We can however  measure some correlation between the communities as in Figure \ref{Intersect_components}.
If $C_{i,t}^1$ be the $i$-th component at time $t$ in $G_1$ and let  $\bar{C}_1=\cup_{i,t} C_{i,t}$, i.e.
 the set of nodes which entered some component at some time. Define the {\em Community Correlation } 
 
$$\rho_C= \frac{|\bar{C}_1 \cap \bar{C}_2|}{max\{  |\bar{C}_1|, |\bar{C}_2|\}} $$
\begin{figure}[ht]  
\begin{center}
 \includegraphics[width=10cm]{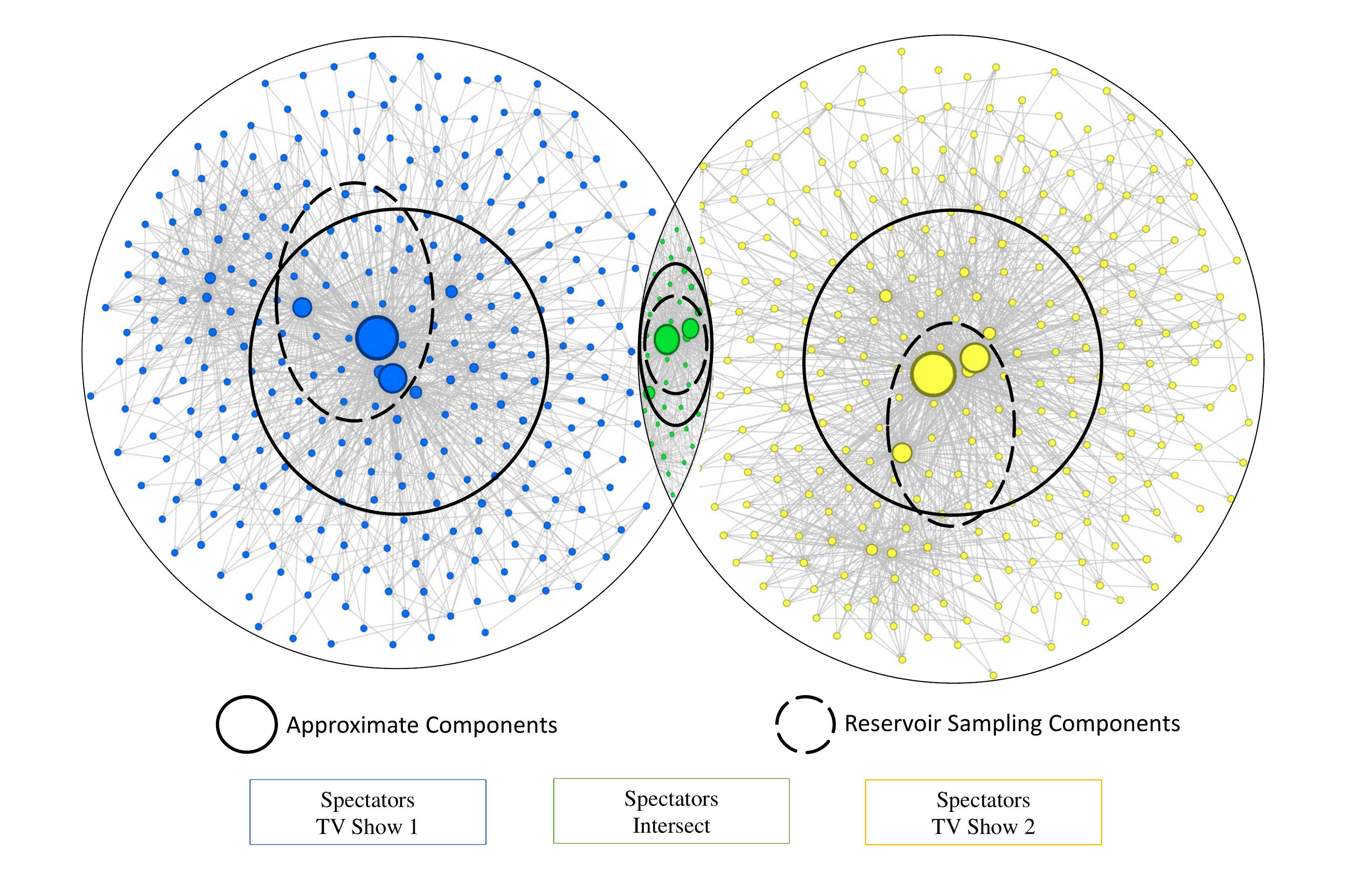}
 \caption{Common communities  between two graphs} \label{Intersect_components}
\end{center}
\end{figure}

We just measure the fraction of nodes in common communities. 
The integration of two streams of edges defining two graphs $G_i=(V_i,E_i)$ for $i=1,2$ can then be viewed as the new structure
$H=(V_1, V_2, V_1\cap V_2, C_1^1,..C_l^1,C_1^2,..C_p^2,\rho_C)$ without edges, where $C_i^j$ is the $i$-th community of $G_j$.
All the sets are exactly or approximately computed from the streams as we stores the nodes and the finite reservoir. It generalizes to $n$ streams as we can look for the correlation of any pair of streams.

Data integration in databases, often studied with data exchange,  does not consider approximation techniques and studies the
schemas mappings. Approximation algorithms, as the one we propose,  give important informations for the integration of multiple sources.

\section{Experiments}

A Twitter stream is defined by a selection: either some set of tags or some geographical position for the sender is given.
A stream of tweets satisfying the selection is then sent in a Json format by Twitter.
We choose a specific tag \#ONPC, associated with  a french TV program which lasts $3$ hours. 
We capture the stream for $4$ hours, starting $1$ hour before the program, and generate the edges as long as they do not contain
\#ONPC.  There are approximately $10^4$ tweets with an average of $2.5$ tags per tweet, i.e. $25.10^3$ potential edges
and $15.10^3$ edges without \#ONPC, whereas there are only $3500$ nodes.
If we do not remove these edges,  the node \#ONPC would dominate the graph and it would not follow our model .

We implemented the Dynamic Community algorithm with the following parameters: $k=400$, $c=3, h=3$, $\tau=15$mins. The nodes are stored in a Mysql database.
The $k$-window reservoir is implemented as a dynamic $k$-reservoir as follows: when edges leave the window, the size of the reservoir decreases. New selected edges directly enter the reservoir when it is not full. When it is full, the new element replaces a randomly chosen element.  This implementation does not guarantee a uniform distribution edges, but is simpler.

Over $4$ hours, there are $16$ intervals for $\tau=15$mins, and $4$ components on the average.The size of a component 
is $8$ on the average. Therefore we store approximately $16*4*8=512$ elements, the representation of the dynamic of the communities.
Figure  \ref{Connected_components_line_chart} shows the evolution of the sizes of the connected components.
\begin{figure}[ht] 
\begin{center}
 \includegraphics[width=10cm]{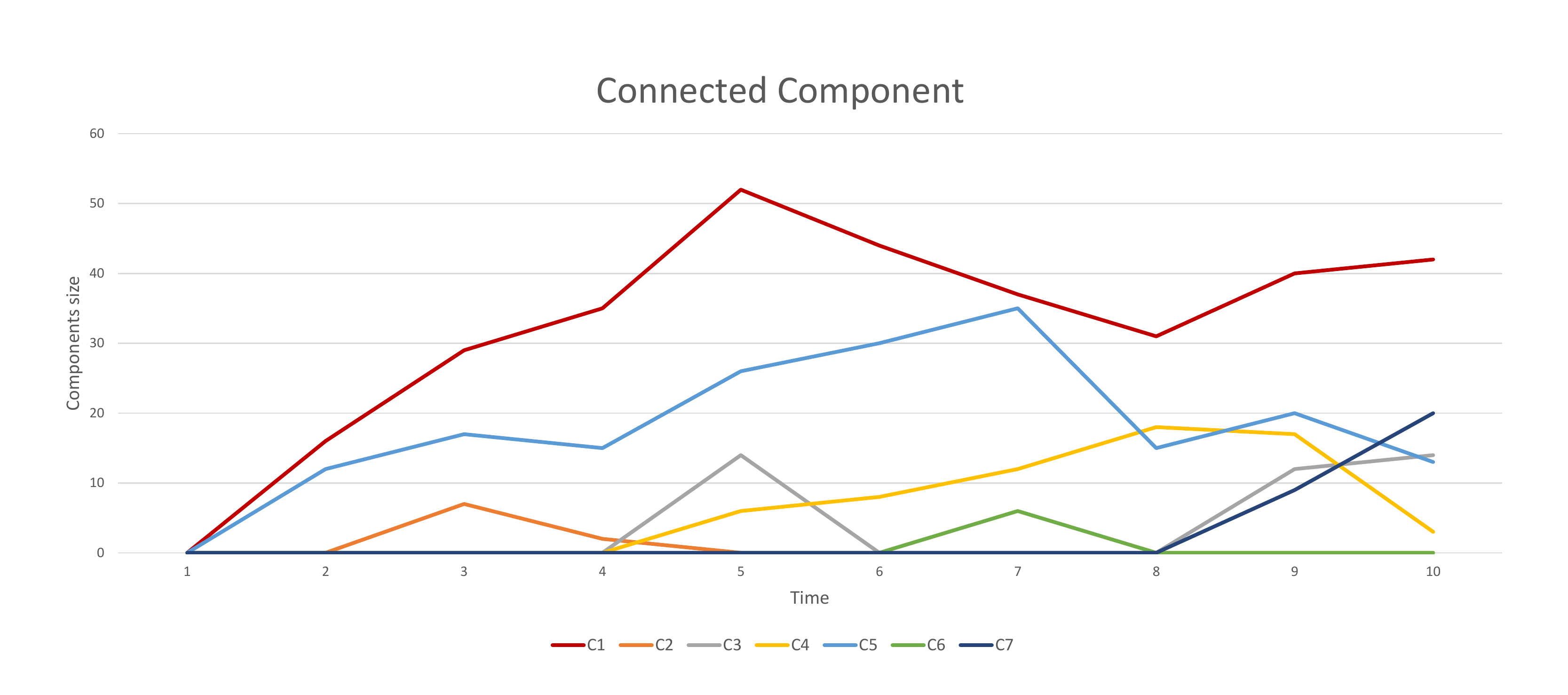}
 \caption{Evolution of the sizes of the connected components} \label{Connected_components_line_chart}
\end{center}
\end{figure}
Each stream can be stored in a compressed form and we can then correlate two streams. We can then compute the Community Correlation. If the two streams have approximately the same length, we can display the correlation online.  The results can be read at {\em http://www.up2.fr/twitter}.

\section{Conclusion}
We presented approximation algorithms for streams of tuples of a Datawarehouse and for streams of edges of a Social graph.
The main DC algorithm computes the dynamic communities of a stream of edges without storing  the edges of the graph and we showed that for concentrated random graphs with $p$ communities whose degrees follow a power law, the algorithm is a good approximation  of the $p$ communities. A finite stream of edges can be compressed as the set of nodes and communities at different time intervals.

In the case of two streams of edges, corresponding to two graphs $G_1$ and $G_2$, we define the Community Correlation of the two streams as the fraction of the nodes in common communities. It is the basis for the Integration of two streams of edges and by extension to $n$ streams of edges. We illustrate this approach with Twitter streams associated with TV programs.\\\\

\newpage

{\Large\bf References\\}

\bibliographystyle{plain}
\bibliography{x1}

\end{document}